\documentclass[12pt]{amsart}
\usepackage{a4wide,palatino,epsfig,latexsym,
amsmath,amsthm,graphicx,hyperref,enumerate,morefloats,color, natbib}
\usepackage{verbatim,caption}
\newtheorem{lemma}{Lemma}
\newtheorem{theorem}{Theorem}

\newtheorem{example}{Example}

\newcommand{\R}{{\mathbb R}}



\begin{document}

\title{Recombination and peak jumping}
\author{ Kristina Crona }

\email{kcrona@american.edu}

\maketitle

\begin{abstract}
We find an
advantage of
recombination
for a category of
complex fitness
landscapes.
Recent  studies  
of empirical fitness
landscapes
reveal complex gene
interactions
and multiple
peaks, and
recombination 
can be a powerful mechanism
for escaping suboptimal
peaks.
 However
classical work
on recombination largely
ignores the effect
of complex 
gene interactions.
The advantage 
we find  has no correspondence
for 2-locus systems 
or for smooth landscapes.
The effect is sometimes extreme, in the
sense that shutting off
recombination could result
in that the organism fails
to adapt. A standard
question about
recombination is
if the mechanism
tends to 
accelerate or 
decelerate 
adaptation.
However,
we argue that 
extreme effects
may be more important
than how the 
majority falls.
\end{abstract}


\section{introduction}


Throughout the paper, we will 
consider haploid biallelic $L$-loci populations.
Let $\Sigma=\{0,1\}$ and let $\Sigma^L$
denote bit strings of length $L$.
$\Sigma^L$ represents the genotype space.
In particular,
\[
\Sigma^2=\{ 00, 10, 01, 11 \}
\text{ and }
\Sigma^3=\{ 000, 100, 010, 001, 110, 101, 011, 111 \}.
\]

We define a fitness landscape as
a function $w:\Sigma^L\mapsto \mathbb{R}$,
which assigns a fitness value to
each genotype \citep{w}. The fitness
of the genotype $g$ is denoted $w_g$.
The metric we consider is the Hamming distance,
meaning that the distance between two genotypes equals
the number of positions where the genotypes differ. 
In particular, two genotypes are adjacent,
or mutational neighbors, if
they differ at exactly one position.
We will use fitness graphs \citep{cgb},
as a representation
of coarse aspects
of fitness landscapes. Roughly, 
the nodes represent genotypes
and each arrow points toward the more fit
genotype. A fitness landscapes is smooth
if it can be represented by a fitness graphs
where all arrows point up.

We are interested in the
effect of recombination for complex fitness
landscapes with multiple peaks.
For general background on recombination,
see e.g. \cite{ol}.
As conventional,
recombination is modeled
so that for a resulting genotype,
each locus is equally likely 
to agree with either parent's allele.

In general, recombination has no effect
for monomorphic populations.
Large population, or subdivided
population, are likely to be 
polymorphic, so that recombination
can produce new genotypes.
One potential advantage of
recombination is the Fisher-Muller effect.
The effect has been discussed
in recent work on complex fitness landscapes
\citep{nns}. Briefly,
 in the absence of recombination,
beneficial mutations may 
get lost due to clonal interference.
For instance, if two single mutants of high fitness co-exist in a population, 
then one of them may outcompete the other.
However, recombination can incorporate beneficial
mutations in the same genome, and thereby prevent
loss of genetic variation due to clonal interference.

\section{Results}
We will study the effect of
recombination 
for an adapting organism.
After a recent change in the environment,
the wild-type, denoted ${\bf{0}}= 0 \dots 0$
does no longer have maximal fitness.
We restrict to landscapes
where there is no path
to the global maximum in
the fitness graph,
so that an adapting
organism could be trapped 
at suboptimal peaks.
At the same time the global maximum
should be within reach, in the
sense that gene shuffling 
of [suboptimal] peaks  
can generate the global maximum.
Specifically, the
fitness landscapes should
satisfy the following 
conditions.

\smallskip
\noindent
{\bf{Main assumptions.}}
Let $g_{\max}$ denote the genotype 
of maximal fitness in $\Sigma^L$.
There exist  genotypes $ g^1, \dots, g^L \in \Sigma^L$,
not necessarily different such that:
\begin{itemize}

\item[(A1)]
There is no 
path from ${\bf{0}}$
to the $g_{\max}$
in the fitness graph.

\item[(A2)]
\[
g_{\text{max}}
=g^1_1 \dots g^L_L,
\]
where $g^k_i$ is the $i:$th bit of $g^k$,

\item[(A3)]
For each $k$,
there exists
a path in the fitness graph 
from ${\bf{0}}$ to $g^k$.

\item[(A4)]
Each $g^k$ is at a local peak.
\end{itemize}

Note that Conditions A1-A3
are important for
our conclusions,
whereas A4 could
be relaxed under
some circumstances.
From a biological perspective,
A1 means that an adapting population
could be trapped at a suboptimal peak,
A2  that the global maximum ($g_{\max}$)
can be obtained by 
a sequence of recombination
events using genotypes in
$\{g^k\}$. Conditions A3 and A4
imply that 
the genotypes
$g^k$  have a reasonable chance
to encounter each other
in large populations.

Landscapes satisfying A1-A4
are of interest, because of the potential
advantage of recombination.
Our first observation is that 
no landscape
satisfies A1-A4 in the 2-loci case.
Indeed, if A1 holds
then the genotype 11 is at a global maximum.
In addition, the single mutants 10, 01
are deleterious. 
In summary
\[
w_{10}, w_{01} < w_{00}   \quad \quad   w_{10}, w_{01} < w_{11}.
\]
Figure 1 shows the fitness graph determined by these conditions,
and clearly  A2-A4 are not satisfied.
However, A1-A4  are satisfied for 
Examples 1 and 2 (see Figures 2 and 3).

\begin{figure}
\begin{center}
\includegraphics[scale=0.5]{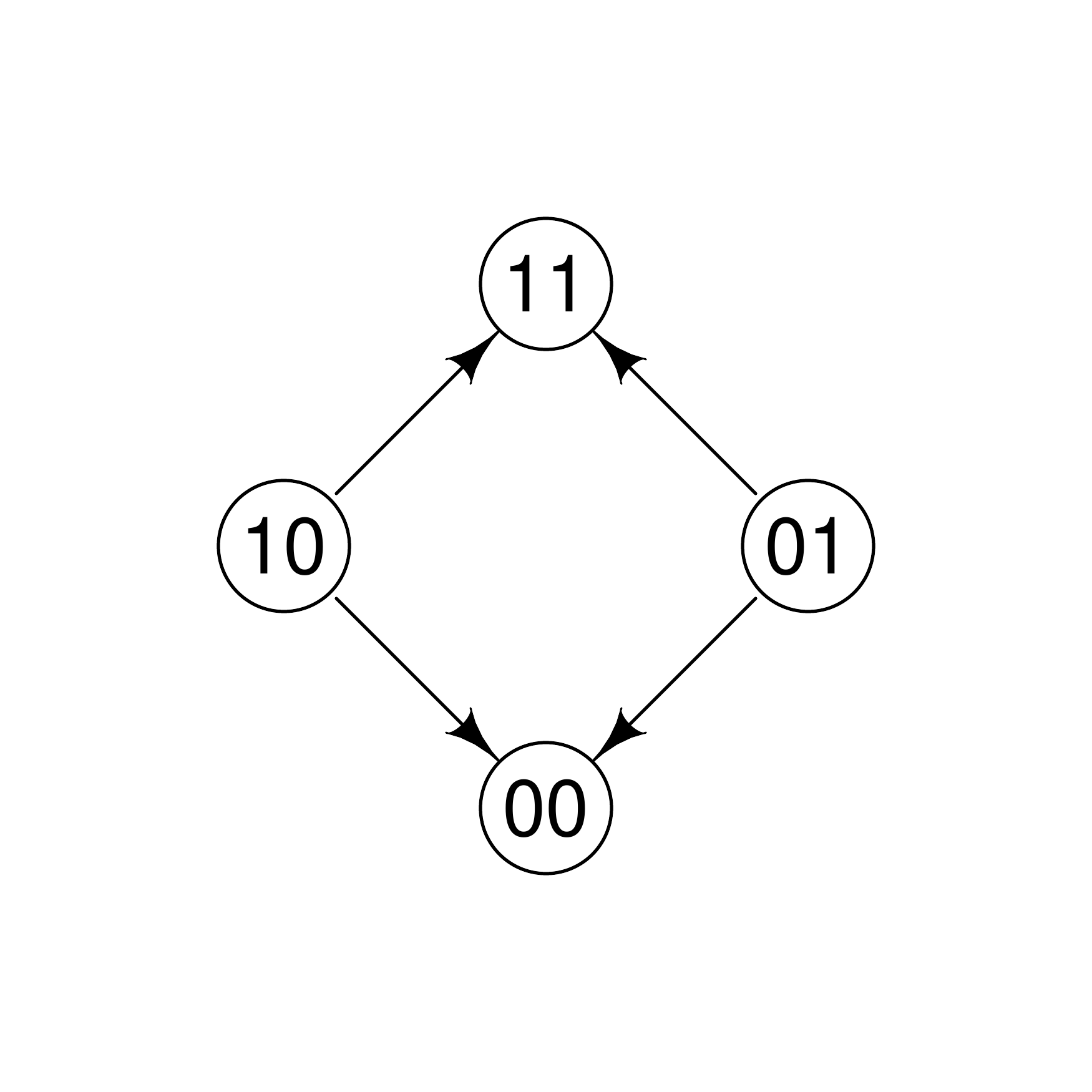}
\end{center}
\caption{The sole fitness graph for 2 loci satisfying condition A1}
\end{figure}

\bigskip
\begin{example}
Consider a system of genotypes
\[
000, 100, 010, 001, 110, 101, 011, 111
\]
The genotypes
$
100, 010, 001
$
have higher fitness than
the wild-type $000$.
The double mutants
$
110,
101,
011
$
have lower fitness
than the wild-type.
The genotype 111 
is at the global maximum.
\end{example}

\begin{example}
For $L=4$, consider
the following case.
The genotypes
1000, 0100, 0010, 0001
have higher fitness than
the wild-type $0000$.
The double mutants
1100,
0011
have higher fitness
than the single mutants.
The remaining double mutants
1010,
0101,
and
the triple mutants
1110,
1101,
1011,
0111
have low fitness.
In summary,
\[
w_{1000}, w_{0100}, w_{0010}, w_{0001}>w_{0000}
\]
\[
w_{1100}>w_{1000}, w_{0100},\quad w_{0011}>w_{0010}, w_{0001},
\]
\[
w_{1111}>w_{1100}, w_{0011},
\]
\[
w_{1001}, w_{1010}, w_{0110}, w_{0101}, w_{1110},  w_{1110} w_{1101} w_{1011} w_{0111}  <w_{0000}.
\]

\end{example}

\begin{figure}
\begin{center}
\includegraphics[scale=0.5]{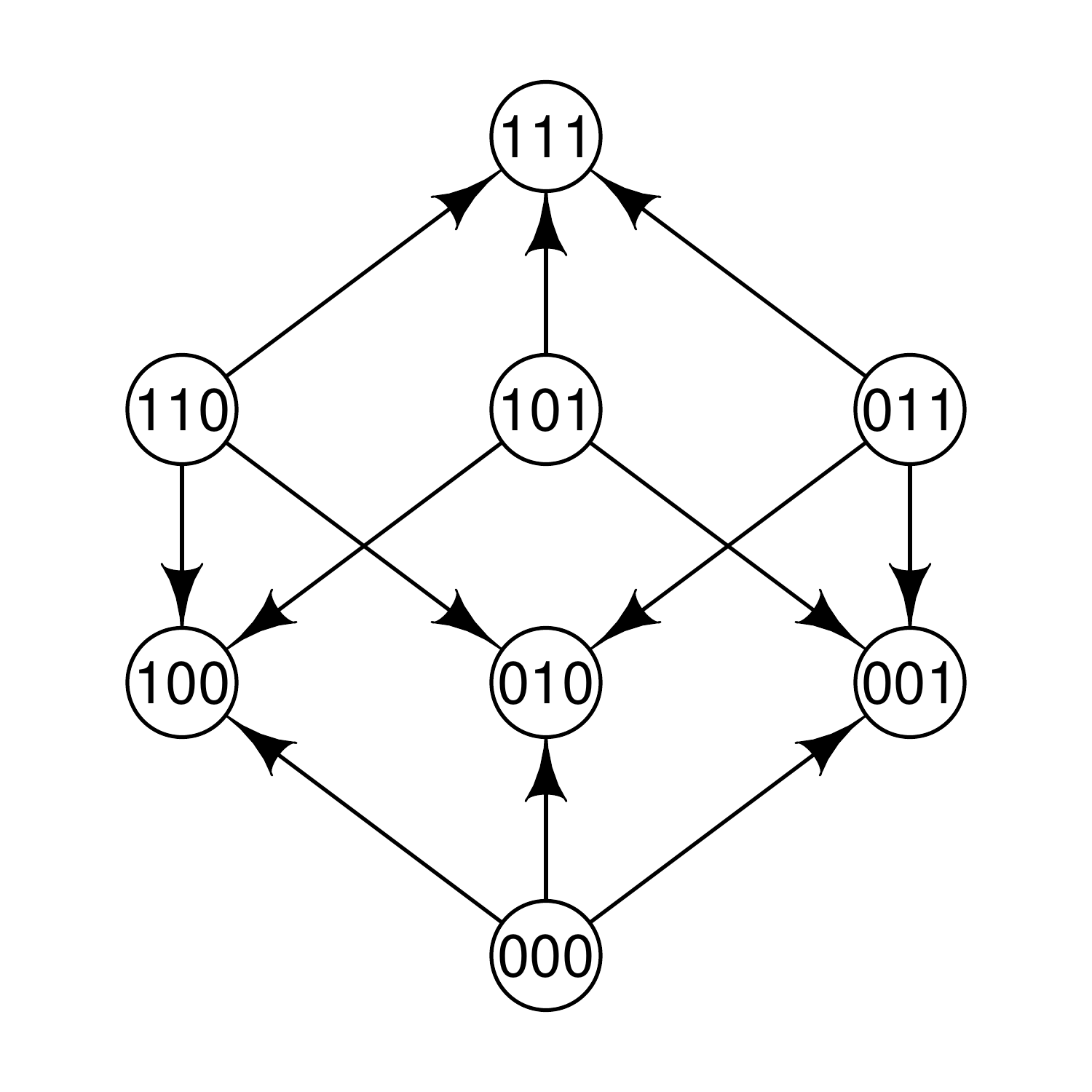}
\end{center}
\caption{A fitness graph for 3 loci. The fitness landscapes
satisfies A1-A4}
\end{figure}

\begin{figure}
\begin{center}
\includegraphics[scale=0.7]{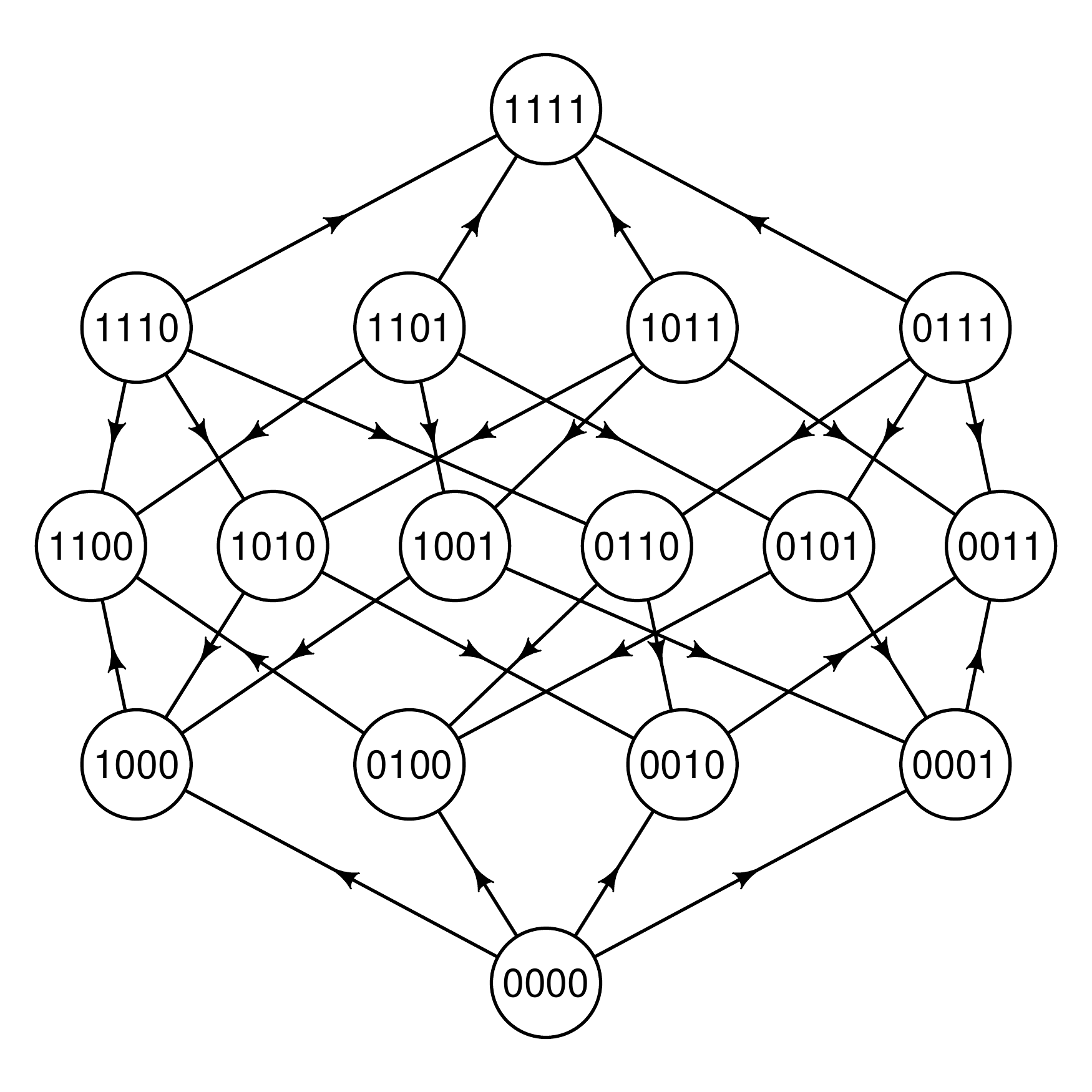}
\end{center}
\caption{A fitness graph for 4 loci. The fitness landscapes
satisfies A1-A4.}
\end{figure}

The double peaked
2-loci case (Figure 1)
and Example 2 (Figure 3) have some
similarities.
Indeed, in both cases there are obstacles
for adaption from the wild-type
to the global maximum.
In  the 2-loci case,
the frequency of 10 and 01 
are expected to be very low,
so that recombination is not
a powerful generator of
the optimal genotype 11. 
In contrast, in Example 2 
both 1100 and 0011 
are peak genotypes,
and recombination may
be a powerful generator
of the optimal genotype 1111.
Informally, one needs to combine "rare and rare"
in the 2 loci case, and "abundant  and abundant"
in Example 2.  

The potential advantage
of recombination 
for an Example 2 population
should not be underestimated.
For instance, consider a relatively
small subdivided population.
Then each subpopulation would be likely to
end up at 1100 or 0011.
Recombination could generate
1111, as soon as there is some
migration between 
the subpopulations.
However, in absence of recombination
the expected time before
1111 appears would be very long.
Indeed, a double mutation
or some other rare scenario
would be necessary. 
A relatively small population,
may fail to produce 1111 genotypes
all together.

One can ask how
frequent fitness landscapes
of the type described in Example 2
are. The TEM-family
of  $\beta$-lactamases
provide an
 interesting example.
 TEM-1 is the wild-type in
 the system, and approximately 200
 mutants have been found clinically, 
 for a record see
 
 http://www.lahey.org/Studies/temtable.asp.

Consider TEM-1, the 4-tuple mutant TEM-50 and 
intermediates. The clinically
found subset of these
16 genotypes are compatible
with the Example 2 fitness graph.
More precisely,
the clinically found alleles
can be represented as:
\[
0000, 1000, 0100, 0010, 0001, 1100, 0011, 1111.
\]
In particular, none of the triple mutants have
been found clinically.
It seems reasonable to interpret the absence of mutants
as an indication of low fitness in a natural setting.
If that interpretation is correct, the set of clinically found mutant is
compatible with the Example 2 fitness graph.

Recombination will 
generate the optimal
genotype for Example 2 
landscapes in many cases.
Other  fitness
landscapes satisfying A1-A4,
may not be quite as favorable.
However, the basic mechanism
is similar.

For instance, consider
Example 1. From 
the three peak genotypes
100, 010, 001
recombination could produce
some  of the intermediates 110, 101, 011,
and in the next step 111. 
For instance, 100 and 010
could produce 110, and then
110 and 001 could produce 111.
Notice that this scenario 
requires that $w_{110}>0$. 
In addition, the population structure would have
to allow for different genotypes to
recombine.

In summary, recombination 
has the potential to generate the optimal
genotypes for A1-A4 landscapes.
As we have seen,
some restrictions on
intermediate genotypes,
and population structure 
and size
may be necessary.
In particular, suboptimal
peak genotypes should have
a chance to recombine.
Notice also, that non-recombining
populations will get
trapped only if genetic
diversity is somewhat
restricted, so that double
mutations are rare.
We restrict to
 A1-A4 landscapes
satisfying the
following conditions.

\noindent
{\bf{Additional conditions:}}
\begin{itemize}
\item
$w_g \neq 0$
for critical
intermediate
genotypes,
\item
the potential generic diversity
is not extreme (double
mutations should still be rare).
\item
The $g^k$ elements
are likely to co-exist, encounter 
each other and recombine
during some stage of
adaptation. 
\end{itemize}
The last condition imposes
constrains on the fitness
landscape as well as the
population structure.
For Examples 1 and 2,
as well as closely
related examples,
the requirement on the
population structure
is modest. 
(Strictly speaking
we do not need A4,
if the last condition is
satisfied in any case.)

We have so far
discussed how recombination
can generate optimal genotypes.
However, even if the optimal
genotype appears in
a population, it is far from
obvious that the genotype 
will go to fixation. If
optimal genotypes recombine poorly,
 they may stay rare in the population.

\medskip
\noindent
{\bf{Main result.}}
Consider
populations which
satisfy A1--A4, as well as the three 
additional assumptions
summarized above.
We argue
that recombination $r>0$ will speed up 
adaptation provided that
$r$ is sufficiently small.

Indeed, in the absence of recombination
a population will  be trapped
at a suboptimal peak for a very
long time.
Rare events, such as double mutations, 
will be necessary for escapes.
However, recombination will
generate the optimal 
genotype within a relatively
short time interval
under our assumptions.
As soon as 
the optimal genotype appears,
Theorem 1 [below] shows
that the proportion is
expected to grow provided
that $r$ is sufficiently small.

\begin{theorem}
For a fitness landscape $w: \Sigma^L \mapsto \R$,
let $g_{\max}$ be the genotype of maximal 
fitness $w_{\max}$, and let $w_{\max'}$
denote the second to highest fitness.
Consider a population with 
genotypes in $\Sigma^L$
and recombination rate $r$.
If $g_{\max}$ is present in a population.
then its proportion
is expected to increase
provided that
\[
r<\frac{ w_{\max}-w_{\max'}} { w_{\max}+w_{\max'}}.
\]
\end{theorem}

The proof of Theorem 1 depends on the following lemma.

\begin{lemma}
For a fitness landscape $w: \Sigma^L \mapsto \R$,
let $g_{\max}$ be the genotype of maximal 
and, and let $r$ denote the
recombination rate. 
Then
$
 | \Sigma^L \setminus \{ g_{\max} \} |
$
will not increase more
than by a factor
$(1+r)$ as a result of recombination.
\end{lemma}

\begin{proof}
Recombining pairs of genotypes are of three types:
\[
\{ g_{\max},  g_{\max} \}, \quad
\{g_{\max},  g : g\neq   g_{\max} \}, \quad
\{ g, g'  : g, g' \neq g_{\max} \}.
\]
In the first case, recombination has
no effect.
In the third case, 
the number of elements in
$
 | \Sigma^L \setminus \{ g_{\max} \} |
$
will not increase.
In the second case,
recombination will results in
at most 2 elements in 
$\Sigma^L \setminus \{g_{\max} \} $.
In summary,
only  the second
case may lead to an
increase of $|\Sigma^L \setminus \{g_{\max} \} |$,
and the maximal net effect is one
1 more element in 
$\Sigma^L \setminus \{g_{\max} \}$.

An increase of
$
|\Sigma^L \setminus \{g_{\max} \}|
$
cannot exceed the case where
all recombining pairs are of the second type,
which translates to an increase by a factor
$(1+r)$.
\end{proof}

\noindent
{\emph{Proof of Theorem 1.}} 
\begin{proof}
The proportions $p_1$ of $g_{\max}$ genotypes,
and $p_2$ of remaining genotypes
after reproduction
can be expressed as
\[
p_1=\frac{C_1  \tilde{p_1}} {  C_1 \tilde{p_1}+ C_2 \tilde{p_2} } , \quad
p_2=\frac{C_2  \tilde{p_2}} {  C_1 \tilde{p_1}+ C_2 \tilde{p_2} },
\]
where $\tilde{p_1}, \tilde{p_2}$ are the former proportions.
From considering $g_{\max}$ individuals which do not
recombine, one concludes that
\[
C_1
\geq 
w_{\max} (1-r)  
\]
Consider $\Sigma^L \setminus \{ g_{\max} \}$.
In order to obtain an upper bound of $C_2$
we need to consider fitness as well as the effect
of recombination.
In total,
\[
C_2
\leq 
w_{\max'}  \, (1+r),
\]
by Lemma 1.
It follows that  the $g_{\max}$ proportion
increases if
\[
w_{\max} \,  (1-r) \,
>  \,
w_{\max'}   \, (1+r) ,
\]
or if
\[
r< \frac{w_{\max}-w_{max'}}{w_{\max}+w_{\max'}}
\]
\end{proof}

\bigskip
One can ask how important it
is that $r$ is small.
We performed
simulation of Example 1
fitness landscapes and large
populations,
using the programming
language R. 
According to our simulations,
the population quickly
goes to fixation at the optimal
genotype for
a wide range of values $r$.
However, the optimal genotype
does not go to fixation
if $r$ is close to 1.
Intuitively, this should make
sense, since $111$ recombines 
poorly with all mutational neighbors 
$110,101,011$. 
The exact threshold for $r$
depends on the choice
of parameters. 
However, the
general pattern is clear.
Recombination is a powerful
mechanism for escaping peaks,
as long as $r$ is sufficiently
small.

One of the most favorable situations 
for A1-A4 populations is probably
subdivided populations with regular
mixing of subpopulations. 
For subdivided populations, chances are good
that all the necessary peak genotypes
are available in the global population.
The fact that recombination
sometimes is especially advantageous for
subdivided populations is well known  \citep[e.g][]{ol}.
One favorable case is the
puddle and flood model \citep{c},
where long periods of isolated adaptation 
for the subpopulations
are alternated by brief periods
of population mixture.

\section{Discussion.}
We have demonstrated
an extreme advantage
of recombination for
a category  of  complex fitness landscapes.
We refer to the mechanism
as "peak jumping", since shuffling  of
peak genotypes generate a new peak
(the jump is from peaks to peak, rather than
from valleys to peak).

The peak jumping effect
should not be confused by
the Fisher-Muller effect mentioned
in the introduction.
The Fisher-Muller effect,
depends on clonal interference,
but the advantage of peak
jumping is different.
For instance, suppose that
some peak genotypes
(such as 1100 and 0011 in Example 2)
have equal fitness, and co-exist
in similar proportions in the population. 
Then clonal interference
is not an issue. 
Indeed, the genetic variation
would be maintained in the population
also in the absence of recombination.
However, because of potential 
peak jumps, 
 (such as the move from 1100 and 0011 to
 1111 in Example 2) 
 recombination may be advantageous.

The peak jumping effect
typically require that
the recombination rate
is relatively small,
and we have provided
a sufficient condition
on the rate.
Our observation that
rare recombination  works
better than frequent 
recombination, agrees
to some extent
with other studies
of complex fitness landscapes 
\citep{nns, me, dpk}.
Infrequent
recombination can be
a powerful mechanism
for escaping peaks
also in cases when 
frequent  recombination 
is not advantageous.

Whether the peak jumping
effect is important or not
is an empirical question.
As mentioned, 
the TEM-family
of  $\beta$-lactamases
provides some empirical
support.
In general,
empirical fitness
landscapes 
are many times
complex with 
multiple peaks
\citep[e.g.][]{h,dk, kk,ssf}.
Notice also that
there are theoretical
arguments why
an adapting 
organism 
would tend to show
more complex
gene interactions
over time \citep{gc,dp}.

A few recent studies
analyze the effect of
recombination
on complex 
fitness landscapes, \citep[e.g.][]{nns, me, mkb, dpk}.
The results point in slightly different directions,
depending on assumptions and
how the problem is phrased. 
Recombination is 
sometimes described as 
a disadvantage. 
However, we  argue that extreme effects
of recombination, such 
as peak jumping,
may be more important
than how the majority falls,
i.e. if recombination more
frequently accelerate
or decelerate adaptation.
The peak jumping effect
observed has no correspondence
in the two-loci case, or for smooth
landscapes. It would be interesting
to further explore effects of recombination
specific for complex fitness landscapes.

Along with results on higher order
epistasis \citep{bps, bpse}, see also \citet{wlw},
our results 
suggest that fitness landscapes
need to be studied in their
full complexity. Arguments based
on pairwise gene interactions or
average curvature may be misleading.

Recombination is wide spread in nature, but
remains poorly understood. 
One of the most important challenges
in the field in our view, is to better understand
the relation between higher order epistasis
and the effect of recombination.

\end{document}